\documentclass[letterpaper, 10 pt, conference]{ieeeconf}  % Comment this line out
\IEEEoverridecommandlockouts                              % This command is only
\usepackage{graphics} % for pdf, bitmapped graphics files
\usepackage{epsfig} % for postscript graphics files
\usepackage{times} % assumes new font selection scheme installed
\usepackage{amsmath} % assumes amsmath package installed
\usepackage{amssymb}  % assumes amsmath package installed
\usepackage{algorithm}
\usepackage{algpseudocode}
\usepackage{url}
\usepackage{float}
\usepackage{color}
\usepackage{subcaption}

\newcommand{\wb}{\mathbf{w}{}}
\newcommand{\zerob}{\mathbf{0}{}}

\newcommand{\vb}{\mathbf{v}{}}
\newcommand{\yb}{\mathbf{y}{}}
\newcommand{\xb}{\mathbf{x}{}}
\newcommand{\Ab}{\mathbf{A}{}}
\newcommand{\Cb}{\mathbf{C}{}}
\newcommand{\nub}{\boldsymbol{\nu}{}}
\newcommand{\etab}{\boldsymbol{\eta}{}}
\newcommand{\Ib}{\mathbf{I}{}}
\usepackage{xcolor}

\newtheorem{assumption}{Assumption}
\newtheorem{theorem}{Theorem}

\title{\LARGE \bf
On the Generalization Properties of Selective State-Space Models for Filtering Tasks for Unknown Systems
}

\author{Alex Tang$^\dag$, M. Emrullah Ildiz$^\dag$, Batin Kurt, Samet Oymak, and Necmiye Ozay% <-this % stops a space
\thanks{$^\dag$Joint first authors with equal contribution}%
\thanks{A.T., M.E.I., S.O., and N.O. are with the Department of Electrical Engineering and Computer Science, University of Michigan, Ann Arbor, USA. B.K. is with the Department of Electrical and Electronics Engineering, Middle East Technical University, Ankara, Turkey, and was visiting University of Michigan when the work was performed. A.T., B.K., and N.O. are funded in part by ARO grant W911NF-26-1-A127 and ONR grant N00014-21-1-2431 (CLEVR-AI). M.E.I. and S.O. are funded in part by NSF grants CCF-2046816, CCF-2403075, and CCF-2212426, and by ONR grant N00014-24-1-2289.}
}

\begin{document}

\maketitle
\thispagestyle{empty}
\pagestyle{empty}

%%%%%%%%%%%%%%%%%%%%%%%%%%%%%%%%%%%%%%%%%%%%%%%%%%%%%%%%%%%%%%%%%%%%%%%%%%%%%%%%
\begin{abstract}
Selective State-Space Models (SSMs) such as Mamba have emerged as an alternative architecture to self-attention based transformers in sequence modeling tasks. Recent works have demonstrated the use of transformers in some filtering and output prediction tasks via in-context learning. In this paper, we analyze whether structured SSMs can work equally well for filtering of unknown systems. In particular, we train the SSM on trajectory samples from a set of systems. At run-time, the SSM is given the outputs of an unknown system from the same set and is expected to predict the next output online. Theoretically, under appropriate assumptions, we derive generalization bounds as to why SSMs succeed in such tasks. Empirically, we demonstrate the performance via several numerical examples. We also discuss the advantages and disadvantages of SSMs versus transformers for this task.

\end{abstract}

%%%%%%%%%%%%%%%%%%%%%%%%%%%%%%%%%%%%%%%%%%%%%%%%%%%%%%%%%%%%%%%%%%%%%%%%%%%%%%%%
\section{INTRODUCTION}

The problem of predicting future trajectories from noisy observations is central to many control applications. When an accurate system model is available \textit{a priori} and the noise satisfies standard assumptions, e.g., Gaussianity, a broad class of estimation methods is available. These include the  Kalman filter~(KF) for linear systems \cite{Kalman1960}, the extended Kalman filter~(EKF) for nonlinear systems \cite{AndersonBook}, and sampling-based methods such as particle filters for nonlinear and non-Gaussian settings \cite{Moral1997}.

Since classical filtering methods require an explicit model of the system dynamics, their direct application is limited when such a model is unavailable or difficult to obtain. In recent years, learned sequence models have become increasingly prominent for prediction and inference tasks \cite{forgione2023from, akram2026transformersimplicitstateestimators}, especially with the widespread adoption of transformer architectures \cite{AttentionTF}. A key advantage of these models is their ability to perform in-context learning~\cite{brown2020language}, namely, to infer useful task-specific structure directly from the input at run-time. This makes sequence models particularly useful for online prediction problems in which the underlying system is unknown and must be inferred implicitly from past observations. Recent work has shown that transformer-based architectures can indeed be trained in this manner for filtering and output prediction of unknown systems \cite{ZheTransformers2023}.

While transformer models have been gaining traction in filtering applications, they are neither the only nor the undisputed best solution for this problem. Interest has grown for an alternative architecture, the State Space Model, which we shall refer to as the SSM. The quadratic-time self-attention modules of the transformer are replaced in SSMs by linear state-space models which have linear time complexity at prediction time \cite{Alonso2025}. A major improvement to SSMs came with the advent of the Mamba \textit{Selective} SSM architecture \cite{Mamba2024}, which uses a linear time-varying (LTV) state-space model instead of earlier time-invariant architectures. While technically a subset of general State-Space Models, further instances of the initialism “SSM” will refer exclusively to the Selective SSM, as it is the architecture of interest in this paper.

In addition to run time advantages, state-space models have other potential strengths over transformers. Transformer models are trained and operate with a specific, finite context window which dictates the longest sequence length that it is able to consider. While this context window can be, in theory, made arbitrarily large, the onerous computational and memory requirements of transformers mean that in practice, it is possible for the context window to be shorter than the desired prediction horizon. There are workarounds to this, such as the sliding window \cite{Beltagy2020}, but these methods are even more hardware-demanding. More importantly in the context of control systems, theoretical guarantees of transformer performance \cite{ZheTransformers2023, Li2023} do not immediately extend to the sliding window implementation. This makes them problematic in applications where prediction accuracy or safety cannot be exposed to risk. SSMs, however, do not suffer from this drawback. This is due to the fact that they have no context window, since they simply apply a recurrent map to the input at each time step. Therefore, no architectural limitations prevent them from working for arbitrarily long prediction horizons, for which any theoretical guarantees on SSM performance are automatically valid.

In practice, SSMs have already been experimentally applied to filtering problems. Empirical results \cite{Liu2026} show that the performance of SSMs in this setting matches, or sometimes even exceeds, the performance of transformer models in a variety of systems exhibiting complex phenomena such as chaotic dynamics. This, combined with prior experimental work on transformer filters, which showed that they already can approach the performance of Kalman filters \cite{ZheTransformers2023} indicates the potential of SSMs in this domain.

In this paper, we focus on generalization performance guarantees for SSMs. While such analyses for transformers have been published in the literature \cite{ZheTransformers2023}, no such bounds currently exist for SSMs. This paper introduces the — to our knowledge — first generalization guarantees for SSMs specific to filtering applications. In Section~\ref{sec:theory}, we will establish and prove such a bound, and in Section~\ref{sec:experiment}, we will present a number of experimental results which showcase the generalization ability of SSMs.

\emph{Specific Notation:} We will use $\mathbb{R}_{+}$ to mean the non-negative real numbers. $\mathbb{R}_{++}$ correspondingly denotes $\mathbb{R}_+ \setminus \{0\}$. 
The symbol $\otimes$ will denote the Kronecker product, while $\times$ denotes multiplication. The $d$-dimensional identity matrix will be referred to as $\mathbf{I}_d$, but we will drop the subscript in situations where the dimension is clear from context. Finally, $\mathrm{softplus}(x) \triangleq \ln(1+e^x)$.

%%%%%%%%%%%%%%%%%%%%%%%%%%%%%%%%%%%%%%%%%%%%%%%%%%%%%%%%%%%%%%%%%%%%%%%%%%%%%%%%
\section{PRELIMINARIES AND PROBLEM SETUP}
\label{sec:prelim}

\subsection{The SSM Architecture}

In this section, we summarize the SSM architecture in \cite{gu2022parameterization,Mamba2024}. The selective SSM with $n_i$ input and $n_o$ output channels is derived from a continuous-time LTV state-space model
\begin{equation}
\begin{aligned}
    \label{eq:mambaltv}
        \dot{\mathbf{h}}(\bar{t}) & = \mathbf{A}{\mathbf h}(\bar{t}) + \mathbf{B}(\bar{t})\mathbf{y}(\bar{t}) \\
        \hat{\mathbf{y}}(\bar{t}) & = \mathbf{C}(\bar{t})\mathbf{h}(\bar{t}),
\end{aligned}\end{equation}
where $\mathbf{h}(\bar{t}) \in \mathbb{R}^l$ is the hidden state of the SSM, $\mathbf{y}(\bar{t}) \in \mathbb{R}^{n_i}$ is the input, and $\hat{\mathbf{y}}(\bar{t}) \in \mathbb{R}^{n_o}$ is the output. The diagonal matrix $\mathbf{A} \in \mathbb{R}^{n_ol\times n_ol}$ and the weight matrices $\mathbf{B}(\bar{t}) \triangleq \mathbf{I}_{n_o} \otimes \mathbf{W}_\mathbf{B} \mathbf{y}(\bar{t})$ and $\mathbf{C}(\bar{t}) \triangleq \mathbf{I}_{n_o}\otimes\mathbf{y}(\bar{t})^\intercal \mathbf{W}_\mathbf{C}^\intercal$ are learnable parameters. 

This continuous-time model is discretized with a non-uniform sampling time defined by $\Delta_{\bar{t}} \triangleq \mathrm{softplus}(p+\mathbf{q}^\intercal\mathbf{y}(\bar{t}))$, again with learnable parameters $p$ and $\mathbf{q}$. Assuming a zero-order hold, this gives a time-varying discrete-time linear system.\footnote{The time $\bar{t}$ of the continuous-time system \eqref{eq:mambaltv} is just a fictitious construct that does not correspond to real time, whereas time-index $t$ in the discrete-time model will be treated as the actual time from now on.} When taking $\mathbf{h}(0)=0$, the unrolled form of this discrete-time dynamics is written as \cite{Mario2025}:
\begin{equation}
\label{eq:ps5}
    \hat{\mathbf{y}}_t = \left(\mathbf{I} \otimes\mathbf{y}_{t-1}^\intercal \mathbf{W}_\mathbf{C}^\intercal\right)\sum_{i=0}^{t-1}\left[\mathbf{A}_\Delta^{t-1-i} \Delta_i\left(\mathbf{I} \otimes \mathbf{W}_\mathbf{B} \mathbf{y}_i \right)\mathbf{y}_i \right],
\end{equation}
where $\mathbf{A}_\Delta^i = \exp\left[(\Delta_{t-1} + \cdots + \Delta_{t-i})\mathbf{A}\right]$ for all $i>0$ and $\mathbf{A}_\Delta^i = \mathbf{I}$ if $i = 0$.

At each time $t$, then, the SSM maps the input sequence $\mathcal{Y}_{t-1} \triangleq \{\mathbf{y}_i\}_{i = 0}^{t-1}$ to the output $\hat{\mathbf{y}}_{t}$. As shorthand for \eqref{eq:ps5}, we write this relation as $\hat{\mathbf{y}}_{t} \triangleq \widehat{\mathrm{SSM}}(\mathcal{Y}_{t-1})$.

Finally, we denote by $\mathcal{A}$ the set of all possible SSMs with this architecture, with parameters defined by the parameter set $\Theta \triangleq \{\mathbf{A},\mathbf{W_B}, \mathbf{W_C},p,\mathbf{q}\}$. All elements of $\Theta$ are assumed to be bounded in the $2$-norm such that $\|\cdot\|_2 < \infty$.

\subsection{Filtering Problem}

Let us specify the filtering problem of interest. Define a distribution $\mathcal{D}_\mathrm{sys}$, from which are drawn systems $\mathcal{S}$ of the form
\begin{align}
    \label{eq:ps1}
    \begin{aligned}
        \mathbf{x}^{\mathcal{S}}_{t+1} &= f_{\mathcal{S}}(\mathbf{x}^{\mathcal{S}}_{t}) + \mathbf{w}_{t+1}\\
        \mathbf{y}^{\mathcal{S}}_{t} &= g_{\mathcal{S}}(\mathbf{x}^{\mathcal{S}}_{t}) + \mathbf{v}_{t}.
    \end{aligned}
\end{align}

The state and output of this system at time $t$ are, respectively, $\mathbf{x}^{\mathcal{S}}_{t}\in\mathbb{R}^n$ and $\mathbf{y}^{\mathcal{S}}_{t}\in\mathbb{R}^m$. $f_{\mathcal{S}}: \mathbb{R}^n\rightarrow\mathbb{R}^n$ and $g_{\mathcal{S}}: \mathbb{R}^n\rightarrow\mathbb{R}^m$ are functions dictating the state dynamics and output, respectively, of system $\mathcal{S}$. The system is further subject to process noise $\mathbf{w}_{t}\sim \mathcal{N}(0, {\sigma_\mathbf{w}^2}\mathbf{I})$ and measurement noise $\mathbf{v}_{t}\sim \mathcal{N}(0, \sigma_\mathbf{v}^2\mathbf{I})$, which are i.i.d. for any system drawn from  $\mathcal{D}_\mathrm{sys}$ and for all time. We use the notation 
$\mathbf{x}_{t}^{\mathcal{S}}=f_{\mathcal{S}}^t(\mathbf{x}^{\mathcal{S}}, \mathcal{W}_{t})$ to denote the $t$-step state evolution of the system from the state $\mathbf{x}^{\mathcal{S}}$ when subject to noise $\mathcal{W}_{t} = \{\mathbf{w}_{i}\}_{i=1}^{t}$.

 We make the following three assumptions on the system defined by \eqref{eq:ps1}:
 
\begin{assumption}
    \label{as:fandgassumptions}
   The system has initial condition $\mathbf{x}_0^\mathcal{S} = \mathbf{0}$. Additionally, $f_\mathcal{S}(\mathbf{0}) = \mathbf{0} = g_\mathcal{S}(\mathbf{0})$.
\end{assumption}

\begin{assumption}
    \label{as:Lipschitz}
    We assume the output is Lipschitz continuous, such that there exists $L_g \in \mathbb{R}_{++}$ such that for all systems $\mathcal{S}$ and for all $\mathbf{x}$, $\mathbf{x}'$, we have $\|g_\mathcal{S}(\mathbf{x}) - g_\mathcal{S}(\mathbf{x}')\| \leq L_g \|\mathbf{x} - \mathbf{x}'\|$.
\end{assumption}

\begin{assumption}
    \label{as:ieiss}
    We assume that the system dynamics are incrementally exponentially input-to-state stable, that is, there exist constants $K \in\mathbb{R}_+$, $\rho\in[0, 1)$, $C_\rho \in \mathbb{R}_{++}$ such that for all $t$,
\begin{multline}
    \label{eq:ps2}
    \|f_{\mathcal{S}}^t(\mathbf{x}^{\mathcal{S}}, \mathcal{W}_{t}) - f_{\mathcal{S}}^t(\mathbf{x'}^{\mathcal{S}}, \mathcal{W}'_{t})\|_2 \\\leq C_\rho\rho^t \|\mathbf{x}^{\mathcal{S}} - \mathbf{x}'^{\mathcal{S}}\|_2 + K\|\mathcal{W}_{t} - \mathcal{W}'_{t}\|_{2, 1},
\end{multline}
where $\|\cdot\|_{2,1}$ means that the $2$-norm is taken for each time element and then the $1$-norm is taken over $t$.
\end{assumption}

\begin{algorithm}[t]
\caption{SSM Training}
\begin{algorithmic}[1]
\Require distribution $\mathcal{D}_\mathrm{sys}$, number of systems $M$, time horizon $T$
\State Sample $M$ systems $\mathcal{S}_k$ from $\mathcal{D}_\mathrm{sys}$
\State From each system $\mathcal{S}_k$, sample a length-$T$ trajectory $\mathcal{Y}^{\mathcal{S}_k}$
%\State $\widehat{SSM} \gets \mathtt{Initialize}$
\State Solve
\begin{equation}
    \label{eq:ps6}
    \widehat{\mathrm{SSM}} = \arg\max_{\mathrm{SSM}\in\mathcal{A}} \frac{1}{MT} \sum_{k=1}^M\sum_{t=1}^T\ell(\mathbf{y}_t^{\mathcal{S}_k}, \mathrm{SSM}(\mathcal{Y}_{t-1}^{\mathcal{S}_k}))
\end{equation}
\Return $\widehat{\mathrm{SSM}}$ 
\end{algorithmic}
\label{alg:main}
\end{algorithm}

Our goal is to train an SSM such that, at run-time, for an arbitrary (unknown) system $\mathcal{S}_0$ drawn from $\mathcal{D}_\mathrm{sys}$, the SSM is fed the system output $\mathcal{Y}^{\mathcal{S}_0}_{t-1}$ and aims to predict $\mathbf{y}_t^{\mathcal{S}_0}$ at each time $t$. In order to solve this filtering problem, we train an SSM using Algorithm~\ref{alg:main}, which solves an empirical risk minimization problem.

As mentioned, once trained, we sample another system $\mathcal{S}_0$ from $\mathcal{D}_\mathrm{sys}$ and require $\widehat{\mathrm{SSM}}$ to predict the next-step output, \textit{viz.}, $\hat{\mathbf{y}}^{{\mathcal{S}_0}}_t = \widehat{\mathrm{SSM}}(\mathcal{Y}_{t-1}^{{\mathcal{S}_0}})$. The use of a diverse set of training data which includes many systems drawn from the distribution $\mathcal{D}_\mathrm{sys}$ is an example of in-context learning and is what permits the model to generalize at run-time to systems not seen during training. This property is critical here, as $\mathcal{S}_0$ is not required to be included in $\{\mathcal{S}_k\}_{k=0}^M$.

Previous work has shown empirically that Mamba can perform in-context learning as well as transformers in decision, regression, and I/O tasks \cite{Park2024}. However, these are limited to experimental results and do not provide rigorous bounds. Additionally, these results come from the setting of sequences where each element is assumed to be i.i.d. — an assumption does not hold for dynamics. This work, then, provides a proof of the SSM's ability to perform in-context learning and generalization in the setting of dynamical systems and therefrom its viability as a filtering algorithm.

%%%%%%%%%%%%%%%%%%%%%%%%%%%%%%%%%%%%%%%%%%%%%%%%%%%%%%%%%%%%%%%%%%%%%%%%%%%%%%%%
\section{THEORETICAL ANALYSIS}
\label{sec:theory}

We seek to show generalization properties of the trained SSM filter, \textit{i.e.}, that given sufficient training data and prediction time horizon, the generalization performance of the SSM filter trained by solving \eqref{eq:ps6} will converge to that of a hypothetical optimal $\mathrm{SSM}^\star$. It is shown in \cite{ZheTransformers2023} that proving this result for sequence models requires two ingredients: 1) the algorithmic robustness of sequence model used — in our case, the SSM — to perturbations in input data and 2) a performance bound for generalization in in-context learning. The former was shown for transformers in \cite{Li2023}, but a similar bound has not been derived for SSMs. Thus, we shall first derive such a result for SSMs in the context of dynamical systems, then show that the generalization bound follows immediately from the work in \cite{ZheTransformers2023}.

\subsection{SSM Robustness}

Consider a system $\mathcal{S}$ satisfying Assumptions~\ref{as:Lipschitz} and~\ref{as:ieiss}. Let $\mathcal{Y}_t$ and $\mathcal{Y'}_t$ be two output trajectories of this system produced by driving the system, respectively, with noise $(\mathcal{W}_t,\mathcal{V}_t)$ and $(\mathcal{W}'_t,\mathcal{V}'_t)$, where $(\mathcal{W}_t,\mathcal{V}_t)$ and $(\mathcal{W}'_t,\mathcal{V}'_t)$ differ only at a time instance $\tau < t$ (\textit{viz.}, $\mathbf{w}_i = \mathbf{w}'_i$ and $\mathbf{v}_i = \mathbf{v}'_i$ for all $i\neq\tau$). As a result, for time steps $i < \tau$, we have $\mathbf{y}_i = \mathbf{y}'_i$, with the outputs differing afterwards. The following theorem establishes the robustness of SSMs to this type of deviations in their input.

\begin{theorem}[SSM Robustness to Perturbation]
\label{th:one} Consider a stable $\mathrm{SSM}\in\mathcal{A}$, \textit{i.e.}, such that the eigenvalues of the continuous-time diagonal $\mathbf{A}$ matrix are strictly negative, applied to predicting the output of a system $\mathcal{S}$.
Let $\mathcal{Y}_t$ and $\mathcal{Y}'_t$ be as previously described. 
Define the event $\mathcal{B} \triangleq \bigcap_{i=0}^t \{\|\mathbf{w}_i\| \leq \Bar{w}, \|\mathbf{v}_i\|\leq\Bar{v}\}$, $\Bar{w}, \Bar{v} \in \mathbb{R}_+$ and write $\mathbb{E}'[\cdot] \triangleq \mathbb{E}[\cdot |\mathcal{B}]$. Define the constants $\bar{y} \triangleq L_g L_\rho\bar{w} + \bar{v}$, where $L_\rho \triangleq \frac{C_\rho}{1-\rho}$, and $\Tilde{y} \triangleq L_gC_\rho\bar{w} + \bar{v}$. Define $\alpha\triangleq\|\exp(\mathbf{A})\|_2^{\mathrm{softplus}(p-\|\mathbf{q}\|_2\bar{y})} < 1$. Let the loss function $\ell(\cdot,\cdot)$ be $L_\ell$-Lipschitz continuous. Then, there exists $K_\mathrm{SSM}\in\mathbb{R}_+$ such that for all $t$,  $\tau < t$, and $\{\mathbf{w}_{0:t-1}$, $\mathbf{v}_{0:t-1}\}$, $\{\mathbf{w}'_\tau$, $\mathbf{v}'_\tau\}$,
\begin{multline}\label{eq:theorem1}
    \mathbb{E}'_{\mathbf{w}_{t}, \mathbf{v}_{t}}\left[ |\ell(\mathbf{y}_t, \hat{\mathbf{y}}_t) - \ell(\mathbf{y}'_t, \hat{\mathbf{y}}'_t)|\right] \\\leq 2L_\ell K_\mathrm{SSM}(t-\tau)^5(\max\{\alpha, \rho\})^{(t-\tau)} \Tilde{y}.
\end{multline}
\end{theorem}

\begin{proof} We start by invoking the Lipschitzness of the loss function as
\begin{equation}
    \label{eq:loss-lip}
    |\ell(\mathbf{y}_t, \hat{\mathbf{y}}_t) - \ell(\mathbf{y}'_t, \hat{\mathbf{y}}'_t)| \leq L_\ell\|\mathbf{y}_t - \mathbf{y}'_t\|_2 + L_\ell\|\hat{\mathbf{y}}_t - \hat{\mathbf{y}}'_t\|_2,
\end{equation}
before proceeding to bound each of the two right-hand terms individually.
\allowdisplaybreaks
\\

\noindent\textbf{Step 1: Bound on $\|\hat{\mathbf{y}}_t - \hat{\mathbf{y}}'_t\|_2$:} From \eqref{eq:ps5}, we first upper bound $\|\hat{\mathbf{y}}_t - \hat{\mathbf{y}}'_t\|_2$. Clearly, all terms with index less than $\tau$ must subtract to zero. Therefore, after pruning zero terms from the summation, we have

\begin{multline}
    \|\hat{\mathbf{y}}_t - \hat{\mathbf{y}}'_t\|_2 \\ = \left\|\left(\mathbf{I} \otimes\mathbf{y}_{t-1}^\intercal \mathbf{W}_\mathbf{C}^\intercal\right)\sum_{i=\tau}^{t-1}\left[\mathbf{A}_\Delta^{t-1-i} \Delta_i\left(\mathbf{I} \otimes \mathbf{W}_\mathbf{B} \mathbf{y}_i \right)\mathbf{y}_i \right]\right. \\ \left. - \left(\mathbf{I} \otimes\mathbf{y'}_{t-1}^\intercal \mathbf{W}_\mathbf{C}^\intercal\right)\sum_{i=\tau}^{t-1}\left[\mathbf{A}_{\Delta'}^{t-1-i} \Delta'_i\left(\mathbf{I} \otimes \mathbf{W}_\mathbf{B} \mathbf{y}'_i \right)\mathbf{y}'_i \right] \right\|_2. \nonumber
\end{multline}

Let us define the following shorthands: $\mathbf{M}_t \triangleq \mathbf{I} \otimes\mathbf{y}_{t-1}^\intercal \mathbf{W}_\mathbf{C}^\intercal$; $\mathbf{N}_t \triangleq \sum_{i=\tau}^{t-1}\left[\mathbf{A}_\Delta^{t-1-i} \Delta_i\left(\mathbf{I} \otimes \mathbf{W}_\mathbf{B} \mathbf{y}_i \right)\mathbf{y}_i \right]$. Then, $\mathbf{M}_t\mathbf{N}_t - \mathbf{M}'_t\mathbf{N}'_t = (\mathbf{M}_t - \mathbf{M}'_t)\mathbf{N}_t + \mathbf{M}'_t(\mathbf{N}_t - \mathbf{N}'_t)$. By the triangle inequality and submultiplicativity,
\begin{multline}
    \label{eq:sr2}
    \|\hat{\mathbf{y}}_t - \hat{\mathbf{y}}'_t\|_2 \leq \|\mathbf{M}_t - \mathbf{M}'_t\|_2\|\mathbf{N}_t\|_2 + \|\mathbf{M}'_t\|_2\|\mathbf{N}_t - \mathbf{N}'_t\|_2.
\end{multline}

Let us further define $\mathbf{G}_{t, i} \triangleq \mathbf{A}_\Delta^{t-1-i}$ and $\mathbf{H}_{t, i} \triangleq \Delta_i\left(\mathbf{I} \otimes \mathbf{W}_\mathbf{B} \mathbf{y}_i\right)$. Then,
\begin{align*}
    \|\mathbf{N}_t - \mathbf{N}'_t\|_2 &\leq \sum_{i=\tau}^{t-1}(\|\mathbf{G}_{t,i} - \mathbf{G}'_{t,i}\|_2\|\mathbf{H}_{t,i}\|_2\|\mathbf{y}_{i}\|_2 \\&+ \|\mathbf{G}'_{t,i}\|_2\|\mathbf{H}_{t,i} - \mathbf{H}'_{t,i}\|_2\|\mathbf{y}_{i}\|_2 \\&+ \|\mathbf{G}'_{t,i}\|_2\|\mathbf{H}'_{t,i}\|_2\|\mathbf{y}_{i} - \mathbf{y}'_{i}\|_2). \nonumber
\end{align*}

Since we want to eventually condition on the event $\mathcal{B}$, we bound the right-hand side of~\eqref{eq:sr2}.
Assume the magnitudes of the process and measurement noise $\mathbf{w}_i$, $\mathbf{v}_i$ are bounded by, respectively, $\Bar{w}$ and $\Bar{v}$. 

From Assumptions~\ref{as:Lipschitz} and~\ref{as:ieiss}, we see that $\|\mathbf{y}_t - \mathbf{y}'_t\|_2 \leq L_g(C_\rho\rho^{t-\tau}\|\mathbf{x}_\tau - \mathbf{x}'_\tau\|_2 + K\|\mathcal{W}_{\tau+1:t} - \mathcal{W}'_{\tau+1:t}\|_{2, 1}) + \|\mathbf{v}_{t} - \mathbf{v}'_{t}\|_{2}$. Moreover, comparing an arbitrary trajectory to a zero-noise trajectory with $\mathbf{x}_0 = \mathbf{0}$, which corresponds to an equilibrium by Assumption~\ref{as:fandgassumptions}, we obtain the bound $\|\mathbf{y}_i\|_2 \leq L_gK\Bar{w} (i-\tau) + \Bar{v}$ and $\|\mathbf{y}'_i\|_2 \leq L_gK\Bar{w} (i-\tau) + \Bar{v}$ for $i\geq \tau$ under the event of bounded noise defined by Theorem~\ref{th:one}. Note that for $i = t$, this bound is maximized. Therefore, we define the function $R_\mathbf{y}(t-\tau) \triangleq L_gK\Bar{w} (t-\tau) + \Bar{v}$, with which we bound the terms in~\eqref{eq:sr2}.

\noindent\textbf{Step 1a: Bound on $\|\mathbf{M}_t - \mathbf{M}'_t\|_2\|\mathbf{N}_t\|_2$:} We begin by using the submultiplicativity of the norm to obtain
\begin{equation}
\begin{aligned}
    \label{eq:MMN1}
    \|\mathbf{M}_t - \mathbf{M}'_t\|_2\|\mathbf{N}_t\|_2 &\leq\|\mathbf{W}_\mathbf{C}\|_2\|\mathbf{W}_\mathbf{B}\|_2\|\mathbf{y}_{t-1} - \mathbf{y}'_{t-1}\|_2\\&\times\sum_{i=\tau}^{t-1}\|\Delta_i\|_2\|\mathbf{A}_\Delta^{t-1-i}\|_2\|\mathbf{y}_i\|_2^2.
\end{aligned}
\end{equation}

As $\Delta_i \triangleq \ln(1+\exp(p+\mathbf{q}^\intercal\mathbf{y}_i))$ and the magnitude of $\mathbf{y}_i$ is bounded by an affine function of $t$, there must exist an affine function of $t$, $\bar{\Delta}(t-\tau)$, such that $\bar{\Delta}(t-\tau) \geq \max_{\tau\leq i<t}\{\Delta_i, \Delta'_i\}$. Then, using $\bar{\Delta}(t-\tau)$ and $R_\mathbf{y}(t-\tau)$, we bound \eqref{eq:MMN1} by
\begin{align*}
    \|\mathbf{M}_t - \mathbf{M}'_t\|_2\|\mathbf{N}_t\|_2&\leq \bar{\Delta}(t-\tau)R^2_\mathbf{y}(t-\tau)\|\mathbf{W}_\mathbf{C}\|_2\|\mathbf{W}_\mathbf{B}\|_2\\&\times\left(\sum_{i=\tau}^{t-1}\|\mathbf{A}_\Delta^{t-1-i}\|_2\right)\|\mathbf{y}_{t-1} - \mathbf{y}'_{t-1}\|_2. \nonumber
\end{align*}

We must now bound the term $\|\mathbf{A}_\Delta^{t-1-i}\|_2 \triangleq \|\exp[(\Delta_{t-1} + \cdots + \Delta_{i+1})\mathbf{A}]\|_2 = \|\exp(\mathbf{A})\|_2^{\Delta_{t-1} + \cdots + \Delta_{i+1}}$. Since $\mathbf{A}$ is a real diagonal matrix with eigenvalues (diagonal elements) $\lambda_i$ strictly negative, the norm $\|\exp(\mathbf{A})\|_2$ is equivalent to $\exp(\lambda_m) < 1$, where $\lambda_m <0$ is the largest eigenvalue of $\mathbf{A}$.

Now, to upper bound $\exp(\lambda_m)^{\sum_{j=i+1}^{t-1}\Delta_{j}}$, we need to lower bound the sum $\sum_{j=i+1}^{t-1}\Delta_{j}$ of step sizes. As $\Delta_i = \mathrm{softplus}(p+\mathbf{q}^\intercal\mathbf{y}_i)$, we have 
\begin{align*}
    \sum_{j=i+1}^{t-1}\Delta_{j} &\geq \sum_{j=i+1}^{t-1}\mathrm{softplus}(p-\|\mathbf{q}\|_2\bar{y}) \\
    &= (t-1-i)\mathrm{softplus}(p-\|\mathbf{q}\|_2\bar{y}) > 0. \nonumber
\end{align*}
If we define $\alpha \triangleq \exp(\lambda_m)^{\mathrm{softplus}(p-\|\mathbf{q}\|_2\bar{y})} < 1$, we obtain $\|\mathbf{A}_\Delta^{t-1-i}\|_2 \leq \alpha^{t-1-i}$.

\noindent\textbf{Step 1b: Bound on $\|\mathbf{M}'_t\|_2\|\mathbf{N}_t - \mathbf{N}'_t\|_2$:} Start with the term $\|\mathbf{N}_t - \mathbf{N}'_t\|_2$. To bound this quantity, we must first bound its components. We examine first $\|\mathbf{G}_{t,i} - \mathbf{G}'_{t,i}\|_2$. Defining $s_{t, i} \triangleq \sum_{j=i+1}^{t-1}\Delta_{j}$ and $s'_{t, i} \triangleq \sum_{j=i+1}^{t-1}\Delta'_{j}$, we proceed by the fundamental theorem of calculus. Noting that $\mathbf{G}_{t,i} - \mathbf{G}'_{t,i} = \int_{s_{t, i}}^{s'_{t, i}} \mathbf{A}\exp(\mathbf{A}u) du$, the normed quantity is then bounded as
\begin{align*}
    \|\mathbf{G}_{t,i} - \mathbf{G}'_{t,i}\|_2 &\leq\int_{\min\{s_{t, i}, s'_{t, i}\}}^{\max\{s_{t, i}, s'_{t, i}\}}\|\mathbf{A}\exp(\mathbf{A}u)\|_2 du \\
    &\leq M_\mathbf{A}\alpha^{t-1-i}\|s_{t,i} - s'_{t,i}\|_2 \\
    &\leq M_\mathbf{A}\alpha^{t-1-i}\|\mathbf{q}\|_2\sum_{j=i+1}^{t-1}\|\mathbf{y}_j - \mathbf{y}'_j\|_2,
\end{align*}
where $M_\mathbf{A}$ is a positive constant of integration and the final inequality is due to the $1$-Lipschitzness of the softplus function. The other terms are straightforward as we have found bounds for them already: $\|\mathbf{G}'_{t,i}\|_2 \leq \alpha^{t-1-i}$; $\|\mathbf{H}_{t,i}\|_2, \|\mathbf{H}'_{t,i}\|_2 \leq \bar{\Delta}(t-\tau)R_\mathbf{y}(t-\tau)\|\mathbf{W}_\mathbf{B}\|_2$; $\|\mathbf{H}_{t,i} - \mathbf{H}'_{t,i}\|_2 \leq (R_\mathbf{y}(t-\tau)\|\mathbf{q}\|_2\|\mathbf{W}_\mathbf{B}\|_2 + \bar{\Delta}(t-\tau)\|\mathbf{W}_\mathbf{B}\|_2)\|\mathbf{y}_i - \mathbf{y}'_i\|_2$. Since $\|\mathbf{M}'_t\|_2 \leq R_\mathbf{y}(t-\tau)\|\mathbf{W}_\mathbf{C}\|_2$, putting it all together, we obtain
\begin{multline}
    \|\mathbf{M}'_t\|_2\|\mathbf{N}_t - \mathbf{N}'_t\|_2 \leq R^2_\mathbf{y}(t-\tau)\|\mathbf{W}_\mathbf{B}\|_2\|\mathbf{W}_\mathbf{C}\|_2\\\times \Bigg[ (R_\mathbf{y}(t-\tau)\|\mathbf{q}\|_2 + 2\bar{\Delta}(t-\tau))\sum_{i=\tau}^{t-1}\alpha^{t-1-i}\|\mathbf{y}_i - \mathbf{y}'_i\|_2 \\ +  \frac{\bar{\Delta}(t-\tau)R_\mathbf{y}(t-\tau)M_\mathbf{A}\|\mathbf{q}\|_2}{1-\alpha}\sum_{i=\tau+1}^{t-1}\alpha^{t-i}\|\mathbf{y}_i - \mathbf{y}'_i\|_2 \Bigg]. \nonumber
\end{multline}
\\

\noindent\textbf{Step 2: Loss function decay bound:} Let us return to the loss function bound in \eqref{eq:loss-lip}:
\begin{multline}
    \label{eq:sr8}
    |\ell(\mathbf{y}_t, \hat{\mathbf{y}}_t) - \ell(\mathbf{y}'_t, \hat{\mathbf{y}}'_t)|
    \leq L_\ell C_\rho\rho^{t-\tau}\|\mathbf{y}_\tau - \mathbf{y}'_\tau\|_2 \\+ L_\ell C_\rho\rho^{t-\tau-1}\bar{\Delta}(t-\tau)R^2_\mathbf{y}(t-\tau)\|\mathbf{W}_\mathbf{C}\|_2\|\mathbf{W}_\mathbf{B}\|_2\\\times\left(\sum_{i=\tau}^{t-1}\alpha^{t-1-i}\right)\|\mathbf{y}_\tau - \mathbf{y}'_\tau\|_2 +  R^2_\mathbf{y}(t-\tau)\|\mathbf{W}_\mathbf{B}\|_2\|\mathbf{W}_\mathbf{C}\|_2 \\\times \Bigg[(R_\mathbf{y}(t-\tau)\|\mathbf{q}\|_2 + 2\bar{\Delta}(t-\tau))\sum_{i=\tau}^{t-1}\alpha^{t-1-i}\rho^{i-\tau}  \\ +  \bar{\Delta}(t-\tau)R_\mathbf{y}(t-\tau)M_\mathbf{A}\|\mathbf{q}\|_2 \sum_{i=\tau+1}^{t-1}\alpha^{t-i}\rho^{i-\tau-1} \Bigg] \\\times\|\mathbf{y}_\tau - \mathbf{y}'_\tau\|_2,
\end{multline}
where the first term comes from applying Assumption~\ref{as:ieiss}
to substitute for $\|\mathbf{y}_t - \mathbf{y}'_t\|_2$ and the second and third terms come from the $\|\hat{\mathbf{y}}_t - \hat{\mathbf{y}}'_t\|_2$ bound derived above.

Set $S_{\alpha}(t-\tau)$ to be the summation $\sum_{i=\tau}^{t-1} \alpha^{t-1-i} = (1-\alpha^{t-\tau})/(1-\alpha)$. Additionally, define the discrete convolution functions $J_{t-\tau}(\alpha, \rho) \triangleq \sum_{i=\tau}^{t-1}\rho^{i-\tau}\alpha^{t-1-i}$ and $K_{t-\tau}(\alpha, \rho) \triangleq \sum_{i=\tau+1}^{t-1}\alpha^{t-i}\rho^{i-\tau-1}$

We now show that the quantity defined by
\begin{multline}
    \label{eq:sr11}
    \sup_{(t-\tau)\in\mathbb{R}_+} \rho^{t-\tau} \\+ \bar{\Delta}(t-\tau)R^2_\mathbf{y}(t-\tau)\|\mathbf{W}_\mathbf{C}\|_2\|\mathbf{W}_\mathbf{B}\|_2S_{\alpha}(t-\tau)\rho^{t-1-\tau} \\+ R^2_\mathbf{y}(t-\tau)\|\mathbf{W}_\mathbf{B}\|_2\|\mathbf{W}_\mathbf{C}\|_2 \\\times \Big[(R_\mathbf{y}(t-\tau)\|\mathbf{q}\|_2 + 2\bar{\Delta}(t-\tau)) K_{t-\tau}(\alpha, \rho) \\ + \bar{\Delta}(t-\tau)R_\mathbf{y}(t-\tau)M_\mathbf{A}\|\mathbf{q}\|_2 J_{t-\tau}(\alpha, \rho)/(1-\alpha)\Big]
\end{multline}
is dominated by a function $\frac{K_\mathrm{SSM}}{C}(t-\tau)^5\rho^{(t-\tau)}$, where $0 \leq K_\mathrm{SSM}/C_\rho < \infty$. Consider two cases. First, if $\alpha\neq\rho$, then
\begin{align*}
    J_{t-\tau}(\alpha, \rho) &= \alpha\frac{\alpha^{t-1-\tau} - \rho^{t-1-\tau}}{\alpha - \rho} \\
    K_{t-\tau}(\alpha, \rho) &= \frac{\alpha^{t-\tau} - \rho^{t-\tau}}{\alpha - \rho}. 
\end{align*}
Also, in the case where $\alpha = \rho$, we have
\begin{align*}
    J_{t-\tau}(\alpha, \rho) &= (t-1-\tau)\rho^{t-1-\tau}\\
    K_{t-\tau}(\alpha, \rho) &= (t-\tau)\rho^{t-1-\tau}.
\end{align*}

To leading order, the products $\bar{\Delta}(t-\tau)R^2_\mathbf{y}(t-\tau)$ and $R_\mathbf{y}^3(t)$ grow as $(t-\tau)^3$, whereas $\bar{\Delta}(t-\tau)R^3_\mathbf{y}(t-\tau)$ grows as $(t-\tau)^4$. Further, $S_\alpha(t-\tau)$ is a constant at leading order. Therefore, the worst-case for the supremum in \eqref{eq:sr11} decays for $(t-\tau) \gtrsim -5/\ln(\max\{\alpha, \rho\})$ at the rate $(t-\tau)^5\max\{\alpha, \rho\}^{(t-\tau)}$ with some constant prefactor, which we set as $\frac{K_\mathrm{SSM}}{C_\rho}$.
\\

\noindent\textbf{Step 3: Bounding the expectation:} Applying this to \eqref{eq:sr8} and taking the expectation on both sides now yields
\begin{multline}
    \mathbb{E}'_{\mathbf{w}_{t}, \mathbf{v}_{t}}[|\ell(\mathbf{y}_t, \hat{\mathbf{y}}_t) - \ell(\mathbf{y}'_t, \hat{\mathbf{y}}'_t)|] \\\leq L_\ell K_\mathrm{SSM}(t-\tau)^5\rho^{(t-\tau)} \mathbb{E}'_{\mathbf{w}_{t}, \mathbf{v}_{t}}\left[\|\mathbf{y}_\tau - \mathbf{y}'_\tau\|_2\right]. \nonumber
\end{multline}

All that remains to be done is to bound the expectation on the right-hand side. By Assumptions~\ref{as:Lipschitz} and~\ref{as:ieiss}, we write $\|\mathbf{y}_\tau - \mathbf{y}'_\tau\|_2 \leq L_gC_\rho\|\mathbf{w}_\tau - \mathbf{w}'_\tau\|_2 + \|\mathbf{v}_\tau - \mathbf{v}'_\tau\|_2$. Taking the conditional expectation then gives $\mathbb{E}'_{\mathbf{w}_{t}, \mathbf{v}_{t}}[\|\mathbf{y}_\tau - \mathbf{y}'_\tau\|_2] \leq 2(L_gC_\rho\bar{w} + \bar{v}) = 2\Tilde{y}$. This fully justifies Theorem~\ref{th:one}.
\end{proof}

\subsection{Filtering Performance Guarantee}
Recall that $\mathcal{A}$ is the space of all trainable SSMs. We are interested in the performance of SSM models in predicting the outputs of the target system $\mathcal{S}_0$ over some time horizon $T \in \mathbb{N}_{++}$. To evaluate this performance, we define the expected risk function
\begin{align}
    \label{eq:fpg1}
    \mathcal{L}(\mathrm{SSM}) \triangleq \frac{1}{T}\sum_{t=1}^T\mathbb{E}\left[\ell(\mathbf{y}^{\mathcal{S}_0}_t, \hat{\mathbf{y}}^{\mathcal{S}_0}_t)\right].
\end{align}

Further, $\mathrm{SSM}^\star \in \mathcal{A}$ represents the hypothetical optimal SSM which minimizes \eqref{eq:fpg1}, \textit{viz.}, $\mathrm{SSM}^\star \triangleq \arg\min_{\mathrm{SSM}\in\mathcal{A}} \mathcal{L}(\mathrm{SSM})$. Define $\widehat{\mathrm{SSM}}$ to be a particular SSM trained on the set of systems $\{\mathcal{S}_i\}_{i=1}^M$ by minimizing the objective function \eqref{eq:ps6}. Now, to evaluate the performance of $\widehat{\mathrm{SSM}}$ \textit{vis-à-vis} $\mathrm{SSM}^\star$, we define the following excess risk function $\mathcal{R}$:
\begin{align}
    \mathcal{R}(\widehat{\mathrm{SSM}}) \triangleq \mathcal{L}(\widehat{\mathrm{SSM}}) - \mathcal{L}(\mathrm{SSM}^\star). \nonumber
\end{align}

Further, we define two pieces of machinery.

\textit{Definition 1 (Distance Metric)}: For any two SSMs in the space $\mathcal{A}$, define the distance metric $\mu(\mathrm{SSM}, \mathrm{SSM}')$ as
\begin{align}
    \mu(\mathrm{SSM}, \mathrm{SSM}') \triangleq \sup_{t\leq T}\sup_{\mathbf{w}_{0:t}, \mathbf{v}_{0:t}}\frac{\|\hat{\mathbf{y}}_t - \hat{\mathbf{y}}'_t\|}{\max_{i< t}\|\mathbf{w}_i\| + \max_{i< t}\|\mathbf{v}_i\|}. \nonumber
\end{align}

\textit{Definition 2 (Covering Number)}: A set $\Bar{\mathcal{A}}_N \triangleq \{\overline{\mathrm{SSM}}_1, \ldots, \overline{\mathrm{SSM}}_N\}$ is an $\varepsilon$-covering of the set of SSMs $\mathcal{A}$ equipped with the distance metric $\mu$ if for all $ \mathrm{SSM}\in\mathcal{A}$, there exists $\overline{\mathrm{SSM}}_i\in\Bar{\mathcal{A}}_N$ such that $\mu(\mathrm{SSM}, \overline{\mathrm{SSM}}_i) \leq \varepsilon$. The covering number $\mathcal{E}(\mathcal{A}, \mu, \varepsilon)$ is the minimum $N\in\mathbb{N}$ such that the $\varepsilon$-covering set $\Bar{\mathcal{A}}_N$ exists.

\begin{theorem}
\label{th:two}
Assume the dynamical system obeys Assumptions~\ref{as:Lipschitz} and~\ref{as:ieiss}, the loss function $\ell(\mathbf{y}_t, \hat{\mathbf{y}}_t)$ is $L_\ell$ Lipschitz and $\ell(\mathbf{y}_t, \cdot)\leq B \in \mathbb{R}_+$. Additionally, Theorem~\ref{th:one} holds with constant $K_\mathrm{SSM}$. Then, when $MT \geq 3\max(\sqrt{n}, \sqrt{m})$, where $n=\dim(\mathbf{x})$ and $m=\dim(\mathbf{y})$, $\forall\varepsilon\in\mathbb{R}_{++}$, for some constant $c$, and with success probability $1-\delta$ or greater,
\begin{align}
    \label{eq:fpg4}
    \mathcal{R}(\widehat{\mathrm{SSM}}) \leq 12B\delta + 4L_\ell\varepsilon + \Bar{B}\sqrt{\frac{\log(4\mathcal{E}(\mathcal{A}, \mu, \varepsilon'))-\log\delta}{cMT}}.
\end{align}

$\Bar{B} = 2B + 7K_{\mathrm{SSM}}L_\ell(L_g C_\rho\sigma_\mathbf{w} + \sigma_\mathbf{v})\sqrt{\log(4MT/\delta)}\frac{T^4\rho^T}{\log\rho}$ and $\varepsilon' = \varepsilon/\left((\sigma_\mathbf{w} + \sigma_\mathbf{v})\sqrt{\log(4MT)-\log\delta}\right)$ are additional derivative constants. Furthermore, $\log(\mathcal{E}(\mathcal{A}, \mu, \varepsilon')) \leq n_\Theta\log(|\Theta|\sqrt{n_\Theta}\varepsilon')$, where $n_\Theta$ represents the number of individual trainable parameters in $\Theta$ and $|\Theta|=5$ the number of parameter blocks.
\end{theorem}

\begin{proof}
The proof of this bound is identical to that of \cite{ZheTransformers2023}, Theorem~\ref{th:one}, except for one modification. We see that \eqref{eq:theorem1} has a similar form as the analogous bound for transformers (\cite{ZheTransformers2023}, Assumption 2) except for the decay rate of the bound in $T$.

This difference is due to the differing decay rate of the algorithmic robustness bound. In the context of transformers, this is as $\mathcal{O}(1/t)$, whereas in Theorem~\ref{th:one}, it is calculated to be $\mathcal{O}(t^3\rho^t)$. This produces a difference only in the martingale difference bound. For transformers, this yields an excess risk bound which decays as $\sqrt{\log^3T/T}$, whereas for SSMs, a decay rate of $\rho^T\sqrt{T^7\log T}$ appears in \eqref{eq:fpg4}. The interested reader may verify that the proof of the aforementioned theorem in \cite{ZheTransformers2023} has no further differences between its application to transformers and SSM. Therefore, we are done.
\end{proof}

We note that the exponential decay of the excess risk means that this bound is significantly stronger for SSMs than it is for transformers. This may explain the superior filtering performance of SSMs in some applications observed in the empirical work.

%%%%%%%%%%%%%%%%%%%%%%%%%%%%%%%%%%%%%%%%%%%%%%%%%%%%%%%%%%%%%%%%%%%%%%%%%%%%%%%%
\section{EMPIRICAL RESULTS}
\label{sec:experiment}
We now evaluate the SSM architecture on online prediction tasks for unknown dynamical systems. We compare its performance with that of a model-free transformer architecture~\cite{ZheTransformers2023} and a classical model-based baseline, namely the standard Kalman filter~(KF). We begin by training the SSM and transformer architectures on trajectories of length $T$ drawn from $M$ randomly sampled systems. 

To evaluate performance, we generate $N$ test systems with trajectory length $T$ and report the root-mean-square (RMS) prediction error, defined as
\begin{align}
\mathrm{RMS}(t)
= \sqrt{\frac{1}{N}\sum_{i=1}^{N}
\left\|\mathbf{y}_{t+1}^{\mathcal{S}_i}-\hat {\mathbf{y}}_{t+1}^{\mathcal{S}_i}\right\|_2^2 } \nonumber
\end{align}
for $t=0,\dots,T-1$, where $\yb_{t+1}^{\mathcal{S}_i}$ denotes the true output of the $i$-th system at time $t+1$, and $\hat \yb_{t+1}^{\mathcal{S}_i}$ denotes the predicted output at the same time.

In the experiments, the SSM architecture uses an internal state dimension of $256$ and an embedding dimension of $512$, for a total of 662 thousand trainable parameters. For the continuous-time diagonal state matrix $\Ab$ in \eqref{eq:mambaltv}, we initialize the largest eigenvalue $s_A$ of $\Ab$ 
to $-0.1$. As a transformer baseline, we consider a GPT-2 architecture~\cite{radford2019language} with model dimension $64$, two attention heads, and two hidden layers, comprising 3.3 million trainable parameters. Both learned models are trained for 50 epochs with a batch size of 64 using the mean-squared-error loss. The implementation is publicly available on GitHub.\footnote{\url{https://github.com/batinkurt1/SSMs_for_Filtering}}

Unless otherwise stated, we train on trajectories of length $T=50$ generated from $M=10000$ randomly sampled systems and evaluate on trajectories of the same length generated from $N=10000$ unseen test systems in the experiments.
\subsection{Linear-Gaussian Setting}\label{sec:lg_setting}
We first consider the linear-Gaussian setting, in which the general nonlinear system in~\eqref{eq:ps1} reduces to
\begin{align}
\begin{aligned}
    \xb_{t+1}^{\mathcal{S}} &= \Ab^{\mathcal{S}} \xb_t^{\mathcal{S}} + \wb_{t}\\
\yb_t^{\mathcal{S}} &= \Cb^{\mathcal{S}} \xb_t^{\mathcal{S}} + \vb_t,
\end{aligned}\nonumber
\end{align}
where $\wb_t \sim \mathcal{N}(\zerob,\sigma_{\wb}^2 \Ib)$ and
$\vb_t \sim \mathcal{N}(\zerob,\sigma_{\vb}^2 \Ib)$.
For each system, the entries of $\Ab^{\mathcal{S}}$ and $\Cb^{\mathcal{S}}$ are generated by sampling each entry independently and uniformly from $[-1.0,1.0]$. As the notation is slightly overloaded, let us clarify to eliminate any confusion that all $\Ab^{\mathcal{S}}$ and $\Cb^{\mathcal{S}}$ in this section refer to \textit{system dynamics} matrices, not SSM parameters as in \eqref{eq:mambaltv}. The matrix $\Ab^{\mathcal{S}}$ is then rescaled such that its spectral radius is $0.95$, and the pair $(\Ab^{\mathcal{S}},\Cb^{\mathcal{S}})$ is required to be observable.

The state and output dimensions are set to $n=5$ and $m=3$, respectively. The process and measurement noise variances are chosen as $\sigma_{\wb}^2=0.01$ and $\sigma_{\vb}^2=0.01$.

\begin{figure*}[!t]
\centering
\begin{subfigure}[t]{0.42\textwidth}
    \centering
    \includegraphics[width=\linewidth]{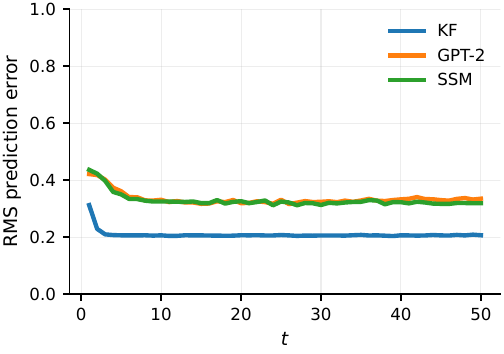}
    \caption{Linear-Gaussian setting}
    \label{fig:linear_case_standard}
\end{subfigure}\hfill
\begin{subfigure}[t]{0.42\textwidth}
    \centering
    \includegraphics[width=\linewidth]{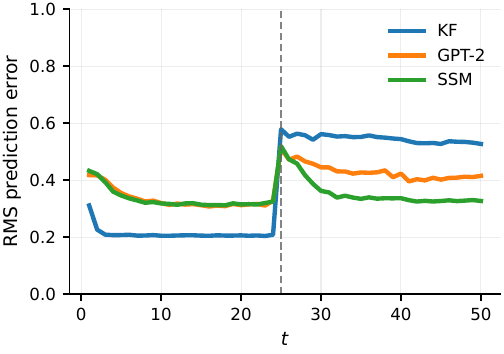}
    \caption{Linear setting with a dynamics switch at $T/2$}
    \label{fig:linear_case_switching}
\end{subfigure}

\medskip

\begin{subfigure}[t]{0.42\textwidth}
    \centering
    \includegraphics[width=\linewidth]{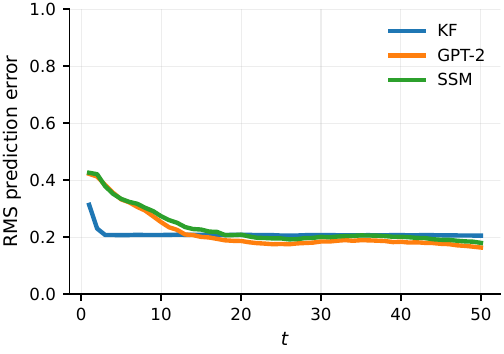}
    \caption{Linear setting with colored noise}
    \label{fig:linear_case_colored_noise_long_correlation}
\end{subfigure}\hfill
\begin{subfigure}[t]{0.42\textwidth}
    \centering
    \includegraphics[width=\linewidth]{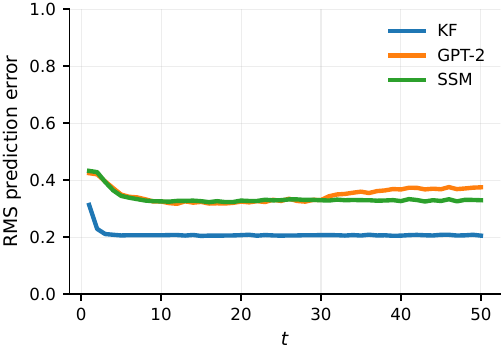}
    \caption{Length generalization}
    \label{fig:linear_case_length_generalization}
\end{subfigure}
\caption{Comparison of RMS prediction errors for the selective SSM, GPT-2, and the Kalman filter in online prediction tasks.} 
\label{fig:linear_system_experiments}
\end{figure*}

The results are shown in Fig.~\ref{fig:linear_case_standard}. As expected, the Kalman filter achieves the lowest prediction error in this linear-Gaussian setting, since it has full knowledge of the underlying model and is optimal for this problem. Although the learned models do not have access to the system matrices, both GPT-2 and the SSM attain stable prediction performance after a short burn-in period. Notably, the SSM achieves this performance with substantially fewer trainable parameters than GPT-2, using 662 thousand parameters compared to 3.3 million for the transformer baseline.

\subsection{Linear Setting with a Dynamics Switch at $T/2$}
In this experiment, the SSM model and the transformer baseline trained in Section~\ref{sec:lg_setting} are tested on linear systems whose dynamics change abruptly at time $T/2$. Specifically, for each test trajectory, the system evolves according to
\begin{align}
    \begin{aligned}
    \xb_{t+1}^{\mathcal{S}} &=
    \begin{cases}
    \Ab_1^{\mathcal{S}} \xb_t^{\mathcal{S}} + \wb_{t+1}, & t < T/2,\\
    \Ab_2^{\mathcal{S}} \xb_t^{\mathcal{S}} + \wb_{t+1}, & t \geq T/2,
    \end{cases}\\
    \yb_t^{\mathcal{S}} &=
    \begin{cases}
    \Cb_1^{\mathcal{S}} \xb_t^{\mathcal{S}} + \vb_t, & t < T/2,\\
    \Cb_2^{\mathcal{S}} \xb_t^{\mathcal{S}} + \vb_t, & t \geq T/2,
    \end{cases}
    \end{aligned} \nonumber
\end{align}
where $(\Ab_1^{\mathcal{S}},\Cb_1^{\mathcal{S}})$ and $(\Ab_2^{\mathcal{S}},\Cb_2^{\mathcal{S}})$ are two distinct observable linear systems generated according to the same procedure as in Section~\ref{sec:lg_setting}. Note that the learned models are trained only on non-switching trajectories, and hence do not observe such abrupt changes during training. We compare the SSM with GPT-2 and the Kalman filter. In this experiment, the Kalman filter is initialized with the first model and continues to use it after the switch, and is therefore mismatched for $t \geq T/2$.

The results are shown in Fig.~\ref{fig:linear_case_switching}. At time $T/2$, all methods exhibit a transient increase in prediction error due to the abrupt change in the underlying system. Since the Kalman filter continues to use the pre-switch model after $T/2$, its error increases significantly and remains high after the switch. In comparison, both learned models adapt to the new dynamics from the observed trajectory, with the SSM recovering more quickly than GPT-2 and attaining a lower post-switch error. This indicates that the SSM is more resilient to abrupt changes in system dynamics.

\subsection{Linear Setting with Colored Noise}
As the next experiment, we consider the same linear system as in Section~\ref{sec:lg_setting}, with the only difference being that the process and measurement noises are now colored. Specifically, we define
\begin{equation}
\wb_t = \frac{1}{\sqrt{15}}\sum_{t'=t-14}^{t} \etab_{t'},\; \vb_t = \frac{1}{\sqrt{15}}\sum_{t'=t-14}^{t} \nub_{t'}, \nonumber
\end{equation}
where $\{\etab_t\}_{t\in\mathbb{Z}}$ and $\{\nub_t\}_{t\in\mathbb{Z}}$ are i.i.d. Gaussian sequences with
    $\etab_t \sim \mathcal{N}(\zerob,\sigma_{\etab}^2 \Ib)$ and
$\nub_t \sim \mathcal{N}(\zerob,\sigma_{\nub}^2 \Ib)$. We take $\sigma_{\etab}^2=0.01$ and $\sigma_{\nub}^2=0.01$. The normalization factor $1/\sqrt{15}$ is introduced so that the marginal variance of the colored noise matches that of the underlying white-noise sequence.

In this setting, the standard Kalman filter is mismatched and is therefore no longer optimal, since the white-noise assumption on the process and measurement noise is violated. The results are shown in Fig.~\ref{fig:linear_case_colored_noise_long_correlation}. The Kalman filter maintains an essentially constant prediction error, effectively treating the noise sequence as white, and thus cannot take advantage of the temporal correlation in the process and measurement noise. By contrast, the learned models are able to capture the temporal correlation and attain a lower prediction error after a short burn-in period. Both GPT-2 and the SSM achieve a lower error than the mismatched Kalman filter after the burn-in period, indicating that the learned sequence models can accommodate temporally correlated process and measurement noise.

\subsection{Length Generalization}
We now examine the ability of the learned models to generalize to trajectories longer than those seen during training. We consider the linear-Gaussian setting of Section~\ref{sec:lg_setting}, and train the SSM and transformer architectures on trajectories of length $T=30$ generated from $M=10000$ randomly sampled systems. The trained models are then evaluated on trajectories of length $T=50$ generated from $N=10000$ unseen test systems according to the same procedure as in Section~\ref{sec:lg_setting}. The difference in train--test sequence length induces a distribution shift and requires the learned models to extrapolate beyond the training horizon.

The results are shown in Fig.~\ref{fig:linear_case_length_generalization}. Up to the training horizon, the SSM and GPT-2 exhibit similar behavior, consistent with the standard linear-Gaussian setting in Fig.~\ref{fig:linear_case_standard}. Beyond the training horizon, GPT-2 shows a clear increase in prediction error, whereas the SSM remains comparatively stable. This indicates that the SSM generalizes more robustly to longer trajectories and is less sensitive to the train--test mismatch in sequence length.


%%%%%%%%%%%%%%%%%%%%%%%%%%%%%%%%%%%%%%%%%%%%%%%%%%%%%%%%%%%%%%%%%%%%%%%%%%%%%%%%
\section{CONCLUSIONS AND FUTURE WORK}

In this paper, the use of selective state space models~(SSMs) for online filtering and output prediction of unknown dynamical systems has been
investigated. In the considered setting, the SSM is trained on trajectory data generated from a family of systems and is then deployed on unseen systems from the same family, using past outputs to predict the next output online. Under suitable assumptions, we derived theoretical generalization bounds for SSMs that help explain why selective SSMs can succeed in this in-context filtering setting. Our empirical results show that selective SSMs are effective for this task across several representative examples. They further indicate that selective SSMs can provide robust performance under certain forms of modeling mismatch and distribution shift.

Future work includes extending the present analysis to broader selective SSM architectures and deriving generalization bounds that explicitly account for distribution shifts. Another important direction is to study multi-horizon prediction in order to characterize the capabilities and limitations of selective SSMs beyond one-step-ahead filtering.

\addtolength{\textheight}{-3cm}

%%%%%%%%%%%%%%%%%%%%%%%%%%%%%%%%%%%%%%%%%%%%%%%%%%%%%%%%%%%%%%%%%%%%%%%%%%%%%%%%
%\section{ACKNOWLEDGMENTS}

%%%%%%%%%%%%%%%%%%%%%%%%%%%%%%%%%%%%%%%%%%%%%%%%%%%%%%%%%%%%%%%%%%%%%%%%%%%%%%%%
\bibliographystyle{IEEEtran}
\bibliography{references}

\end{document}